\documentclass[a4paper]{amsart}

\usepackage{hyperref}
\usepackage{color}
\usepackage{graphicx}
\usepackage{datetime}
\usepackage[margin = 1in]{geometry}
\usepackage{subcaption}

\usepackage{url}

\usepackage{dsfont}
\usepackage{natbib}

\newtheorem{theorem}{Theorem}
\newtheorem{lemma}[theorem]{Lemma}
\newtheorem{proposition}[theorem]{Proposition}

\begin{document}
\title{On Merton's Optimal Portfolio Problem with Sporadic Bankruptcy for Isoelastic Utility}

\author{Yaacov Kopeliovich \and Michael Pokojovy \and Julia Bernatska}
\address{University of Connecticut, Finance Department, 2100 Hillside Road, Storrs, CT 06269 \and
Old Dominion University, Department of Mathematics and Statistics, 4300 Elkhorn Ave, Norfolk, Virginia 23529 \and
University of Connecticut, Mathematics Department, 1 University Place, Stamford, Connecticut 06901, USA}
\email{mpokojovy@odu.edu}
\date{\today}

\thanks{
Helpful discussions with Professor Ramon van Handel on the verification theorem for classical Merton problem with logarithmic utility are greatly appreciated.}

\allowdisplaybreaks[4]
\begin{abstract}
    We consider a stock that follows a geometric Brownian motion (GBM) and a riskless asset continuously compounded at a constant rate. We assume that the stock can go bankrupt, i.e., lose all of its value, at some exogenous random time (independent of the stock price) modeled as the first arrival time of a homogeneous Poisson process. For this setup, we study Merton's optimal portfolio problem consisting in maximizing the expected isoelastic utility of the total wealth at a given finite maturity time. We obtain an analytical solution using coupled Hamilton-Jacobi-Bellman (HJB) equations. The optimal strategy bans borrowing and never allocates more wealth into the stock than the classical Merton ratio recommends. For non-logarithmic isoelastic utilities, the optimal weights are non-myopic. This is an example where a realistic problem, being merely a slight modification of the usual GBM model, leads to non-myopic weights. For logarithmic utility, we additionally present an alternative derivation using a stochastic integral and verify that the weights obtained are identical to our first approach. We also present an example for our strategy applied to a stock with non-zero bankruptcy probability. \\[0.1in]
\end{abstract}

\maketitle

\section{Introduction}

On 3/10/2023, Silicon Valley Bank, one of the largest banks in Silicon Valley, collapsed after a bank run. The collapse sent a shock wave through the financial services in Silicon Valley. This was not the only event of its kind. During the 2007-2008 financial crisis, bank failures were quiet frequent, the most famous ones being the collapse of Lehman Brothers and Bear Stearns that was sold to JPMorgan Chase in March 2008. Such collapses prompt a natural question whether the possibility of a crash can be incorporated into portfolio allocation decisions, and that is what lies at the heart of this paper, i.e., we investigate the dynamic optimal allocation problem in a situation where security prices are subject to exogenous events like bankruptcy. Recall, the dynamic optimal portfolio problem was first considered by~\cite{Mer1969} for multiple stocks that follow a multi-dimensional geometric Brownian motion~\citep{PlaRe2009}.

For a single stock, whose price $S_{t}$ at time $t$ is modeled as a Markovian diffusion process that solves the stochastic differential equation
\begin{equation}
    \label{EQUATION:STOCK_DYNAMICS_GBM}
    \mathrm{d} S_{t} = \mu S_{t} \mathrm{d}t + \sigma S_{t} \mathrm{d}B_t
\end{equation}
with a standard Wiener process $(B_t)_{t \geq 0}$,
while the risk-free asset follows the ordinary differential equation
\begin{equation}
    \label{EQUATION:BOND}
    \mathrm{d} R_{t} = r R_{t}
\end{equation}
with a constant risk-free rate $r$.
Aiming to maximize the expected isoelastic utility
\begin{equation}
    \label{EQUATION:ISOELASTIC_UTILITY_MERTON}
    U(W) =
    \begin{cases}
        \frac{W^{1 - \gamma} - 1}{1 - \gamma}, & \gamma \geq 0, \gamma \neq 1, \\
        \log(W), & \gamma = 1
    \end{cases}
\end{equation}
of the terminal wealth for given $\gamma > 0$, \cite{Mer1969} derived his Merton ratio $\pi$ of the wealth to invest in stock
\begin{equation}
    \label{EQUATION:MERTON_RATIO_CLASSICAL}
    \pi = \frac{\mu - r}{\gamma \sigma^2},
\end{equation}
while the remaining fraction, $(1 - \pi)$, of investor's wealth at each time $t$ is invested in the risk-free asset.
The theoretical importance of Merton ratio is that it gives investors a simple, yet rigorous rule how to allocate capital into the risky asset as to maximize the expected utility, while the practical importance is highlighted in~\cite{Bu} article.

Since the groundbreaking work of~\cite{Mer1969}, many authors investigated various variants of Merton optimal portfolio problem. For example, \cite{Wa2002} considered the optimal portfolio problem for risky assets following a mean reverting processes and obtained an exact solution for the optimal allocation when markets are complete. \cite{Mer1973} himself applied his dynamic portfolio optimization approach to extend the usual capital asset pricing model (CAPM) to an intertemporal CAPM. The latter generalizes the usual CAPM model developed by Sharpe based on Markowitz portfolio problem. Fixed income securities were incorporated as well. For instance, \cite{BSL1997} considered a Merton-type problem for fixed income securities and stocks but without any bankruptcy assumptions on the bonds. \cite{KLSS1985} considered the optimal portfolio problem for a general utility function and incorporated bankruptcy events. The bankruptcy event considered in the latter paper was not modeled as an exogenous event but was rather assumed to arise from the self-exciting (autoregressive) dynamics, i.e., the authors asked the question when it will be optimal for the security to default given an infinite time horizon and solved this problem of optimal bankruptcy. For applications in corporate finance, see \cite{Le1994, LeGo2001}.

The goal of this paper is to revisit Merton's original problem and investigate the optimal allocation between stock and cash as to optimize the expected power utility at a finite maturity $T$ assuming, however, the stock can sporadically jump into bankruptcy at an exogenous random time $\tau \sim \mathrm{Exp}(\lambda)$, independent of the stock price, which is modeled as an exponential random variable with some parameter $\lambda > 0$ corresponding to default intensity. This feature is ubiquitous when investing in stocks of high-yield companies that have non-zero bankruptcy probability at or before a finite maturity time. In fact, we were led to analyze this problem following a discussion with a hedge fund whose specialty was to invest in such securities. As we have two regimes in our dynamical allocation problem, we propose a relatively new technique to address such problems by breaking down the Bellman problem into pre- and post-default regimes. Modifying the It\^{o}'s lemma for such occasion, we obtain an HJB system of two coupled PDEs that we are able to solve analytically. To the best of our knowledge, this is the first time that an analytical solution was obtained for problem involving stocks with a possibility of bankruptcy over a finite time horizon.

The situation we consider is principally different from the usual Merton's problem. One significant difference is that a stock with non-zero bankruptcy probability necessitates that the allocation weight to this stock is always less than 100\% and, thus, borrowing is not allowed for such stocks. This is an interesting feature which is a by-product of our model and is not present in other Merton-type problems encountered in the literature. Another important difference is that our analytical solution for the weights implies non-myopic weights prior to the bankruptcy for any isoelastic utility except for the logarithmic case. In the latter scenario, our solution reduces to a piecewise constant weight given by the formula
\begin{equation}
    \label{EQUATION:MERTON_RATIO_OUR}
    \pi_{t} = \left\{
    \begin{array}{cl}
        \pi_{\mathrm{pre}}^{\ast}, & t < \tau \\
        0, & t \geq \tau
    \end{array}
    \right.
    \quad \text{ with } \quad
    \pi_{\mathrm{pre}}^{\ast} = \frac{\mu - r + \sigma^2 - \sqrt{\left(\sigma^2 - (\mu - r)\right)^2 + 4\lambda\sigma^2}}{2\sigma^2},
\end{equation}
where the pre-default weight $\pi_{\mathrm{pre}}^{\ast}$ neither depends on time nor wealth.
Further, $\pi_{\mathrm{pre}}^{\ast}$ never exceeds the classical Merton ratio, which is also the cae for non-logarithmic utilities. As mentioned above, the investor is discouraged from allocating all wealth into stock or borrowing money to purchase stock. This qualitative behavior is intuitively expected in view of a non-zero default probability.

The method we employ appears to be amenable to generalizations for portfolio selection with regime-switching~\citep{AzPiWe2014, WeShZh2020} with or without jumps governed by semi-Markov processes~\cite[Chapter~7]{Me1984}. These are a natural extension of the market model we examine here. However, we chose to focus on a relatively simple ``prototype'' model to exemplify how our approach works instead of going for redundant technical complication.
A similar problem was recently studied in a different context with engineering applications~\citep{LuLiMaGoNa2023}. However, the authors neither presented an analytical solution nor allowed for random regime changes. Turning back to finance literature, our work also has connection to portfolio optimization driven by a Brownian motion perturbed by jump processes. For example, \cite{EK} investigated optimal portfolios under Lev\'{y} innovations in the Markowitz framework of mean variance.
Their optimization approach followed the classical Markowitz approach and not the Merton approach we present here. Another distinction is that jumps under a Lev\'{y} process have an emerging property while stocks (contrary to bonds) do not recover from bankruptcy.

A dynamic portfolio optimization problem for a stochastic model with jumps was also considered by~\cite{F2014}. Like other works, the processes adopted were emerging processes with regime switching regulated by a Markov control process, which is separate from our case, where the process enters an absorbing state. Another significant work in this direction is due to~\cite{SH2015}, who consider asset allocation when the stock follows a geometric Brownian motion and a jump process. Even though the stochastic processes considered are complicated, the authors were able to obtain an exact solution to the Merton allocation problem for logarithmic utility. The analytical solution they have obtained assumes an infinite time horizon and not the finite maturity we assume in this work.

On the bond side, this type of problem was first studied by~\cite{BJ}. The framework they considered was an optimization problem for stock cash and a bond with some default probability. The authors formulated the underlying optimization problem and presented a solution but omitted the final derivation of the HJB equation or a stochastic integral we employ in our paper.
The issue was later resolved by \cite{CFL} and \cite{CL}, who studied an optimal allocation problem for a bond with bankruptcy and a stock using a method similar to ours. But the problem we discuss is different as we consider stocks with exogenous bankruptcy occurrence.

Our paper is structured as follows. In Section~\ref{SECTION:PRELIMINARIES}, we introduce a probabilistic framework and give a problem statement. In Section~\ref{SECTION:HEURISTIC_DERIVATION}, a simple heuristic calculation is presented to derive the optimal stock wait for the Merton problem with bankruptcy.
Without being backed by a rigorous verification theorem, it is new and interesting on its own merit since it employs a non-conventional form of the It\^{o}'s lemma to formally derive the optimality conditions following the idea originally presented by~\cite{KOP2023}. Similar to \cite{AzPiWe2014, CFL}, a pair of coupled HJB equations is derived to find the optimal weights. The non-logarithmic case presents an interesting technical feature as we obtain an ODE for the weights directly and, after solving it analytically, we derive the value function. As mentioned above, the weight is non-myopic. We also derive the result for the logarithmic utility using an alternative stochastic integral technique.
In Section~\ref{SECTION:RIGOROUS_DERIVATION}, a solution to the optimal portfolio problem is obtained by reducing the expected terminal logarithmic utility objective for a stock with bankruptcy to a combination of exponentially weighted expected terminal and running logarithmic utilities without bankruptcy. The optimal Merton ratio (viz.~Equation~\eqref{EQUATION:MERTON_RATIO_OUR}) obtained with both approaches agree. In Section~\ref{SECTION:EXAMPLE}, we illustrate our approach by computing the optimal allocation for a real stock with a non-trivial bankruptcy probability and confirm a material difference from the allocation produced with the usual Merton ratio. Proofs of auxiliary results are relegated to Appendix~\ref{APPENDIX:AUXILIARY_RESULTS}.

\section{Probabilistic Setup and Preliminaries}
\label{SECTION:PRELIMINARIES}

At time $t \geq 0$, let $S_{t}$ and $R_{t}$ denote the price of a defaultable stock and a riskfree asset, respectively.
Until a default time $\tau$, $S_{t}$ follows a GBM model with drift $\mu$ and volatility $\sigma > 0$, driven by a standard Brownian motion $(B_{t})_{t \geq 0}$, while $R_{t}$ is continuously compounded at a constant rate $r > 0$.
Assuming $W_{t}$ is the total wealth at time $t$ consisting of $\pi_{t} W_{t}$ dollars in stock and $(1 - \pi_{t}) X_{t}$ dollars in the riskless asset for some $\pi_{t}$, the dynamics of $X_{t}$ is given by equation
\begin{equation}
    \label{EQUATION:SDE_WEALTH_WITH_BANKRUPTCY}
    \begin{split}
        \mathrm{d}W_{t} &=
        \left\{
        \begin{array}{cl}
            \big(\mu \pi_{t-} + r (1 - \pi_{t-})\big) W_{t-} \mathrm{d}t + \sigma \pi_{t-} W_{t-} \mathrm{d}B_{t} &\quad\text{for } 0 \leq t < \tau, \\
            r(1 - \pi_{t}) W_{t} &\quad\text{for } \tau \leq t \leq T,
        \end{array}
        \right. \\
        W_{\tau} &= (1 - \pi_{\tau-}) W_{\tau-} \quad \text{if } \tau \geq T, \\
        W_{0} &= w
    \end{split}
\end{equation}
where $w > 0$ is a (given) initial wealth.
Introducing the jump operator~\citep{LuLiMaGoNa2023}
\begin{equation}
    \label{EQUATION:JUMP_OPERATOR}
    \Delta W_{t} := W_{t} - W_{t-}
\end{equation}
and assuming $W_{t}$ and $\pi_{t}$ are c\'{a}dl\'{a}g, the jump condition in Equation~\eqref{EQUATION:SDE_WEALTH_WITH_BANKRUPTCY} is equivalent with
\begin{equation}
    \label{EQUATION:JUMP_CONDITION}
    \Delta W_{\tau} = -\pi_{t} W_{\tau-}.
\end{equation}
Recall that $\tau \sim \mathrm{Exp}(\tau)$ is assumed stochastically independent of $(B_{t})_{t \geq 0}$. To emphasize the dependence of $W_{t}$ on $\pi_{t}$ and $\tau$, we will often write $W_{t}^{\pi, \tau}$.

Consider the semi-Markov process
\begin{equation}
    \label{EQUATION:SEMIMARKOV_SWITCHING_PROCESS}
    \xi_{t} := \mathds{1}_{\{t < \tau\}}
\end{equation}
with the non-absorbing pre-default state $1$ and the absorbing post-default state $0$. Introducing the drift and volatility functions
\begin{equation}
    \begin{split}
        a(t, W, \pi, \xi) &= \big(\mu \pi + r (1 - \pi)\big) \xi +  r (1 - \pi) (1 - \xi), \\
        b(t, W, \pi, \xi) &= \sigma \pi \xi,
    \end{split}
\end{equation}
Equation~\eqref{EQUATION:SDE_WEALTH_WITH_BANKRUPTCY} can be expressed as
\begin{equation}
    \label{EQUATION:SDE_WEALTH_WITH_BANKRUPTCY_TRANSFORMED}
    \begin{split}
        \mathrm{d}W_{t} &=
        a(t-, W_{t-}, \pi_{t-}, \xi_{t-}) \mathrm{d}t + b(t-, W_{t-}, \pi_{t-}, \xi_{t-}) \mathrm{d}B_{t} - \pi_{t-} W_{t-} \mathrm{d}H(t) \text{ for } 0 \leq t \leq T, \\
        W_{0} &= w.
    \end{split}
\end{equation}
with the Heaviside function $H(\cdot) = \mathds{1}_{(-\infty, 0]}(\cdot)$.
It should be emphasized that, unlike~\cite{LuLiMaGoNa2023}, the switching time in Equation~\eqref{EQUATION:SDE_WEALTH_WITH_BANKRUPTCY_TRANSFORMED} is random. Also, unlike~\cite{AzPiWe2014}, the solution process is not pathwise continuous.

Consider the filtered probability space $(\Omega, (\mathcal{F}_{t})_{t \geq 0}, \mathbb{P})$ with the minimal $\sigma$-algebra
\begin{equation}
    \label{EQUATION:SIGMA_ALGEBRA_B_T_AND_XI_T}
    \mathcal{F}_{t} = \sigma(B_{t}, \xi_{t}) \quad \text{ for } t \geq 0
\end{equation}
completed by null sets and assume $(\pi_{t})_{t \geq 0} $ is an adapted c\'{a}dl\'{a}g process with
\begin{equation*}
    \int_{0}^{T} \pi_{t}^{2} \mathrm{d}t < \infty \quad \mathbb{P}\text{-a.s.}
\end{equation*}
A strong solution to the jump SDE in Equation~\eqref{EQUATION:SDE_WEALTH_WITH_BANKRUPTCY_TRANSFORMED} is an adapted c\'{a}dl\'{a}g process $(W_{t})_{t \in [0, T]}$ that satisfies the integral equation
\begin{equation*}
    \begin{split}
        W_{t} &= w + \int_{0}^{t} a(s-, W_{s-}, \pi_{s-}, \xi_{s-}) \mathrm{d}s +
        \int_{0}^{t} b(s-, W_{s-}, \pi_{s-}, \xi_{s-}) \mathrm{d}B_{t}
        - \int_{0}^{t} \pi_{s-} W_{s-} \mathrm{d}H(s - \tau)
    \end{split}
\end{equation*}
where $H(\cdot) = \mathds{1}_{(-\infty, 0]}(\cdot)$ is the Heaviside function. While pathwise uniqueness follows with the standard ODE theory applied pre- and post-jump, the existence is given by Equation~\eqref{EQUATION:STATE_DEPENDENT_CONTROL} below.

Adopting an isoelastic utility function (viz.~Equation~\eqref{EQUATION:ISOELASTIC_UTILITY_MERTON}),
consider the terminal expected utility functional
\begin{equation}
    \label{EQUATION:EXPECTED_TERMINAL_UTILITY}
    J(\pi) = \mathbb{E}\big[U(W^{\pi, \tau}_{T})\big].
\end{equation}
Similar to~\cite{AlPe2021} and \cite{F2014}, the set of admissible controls for Equation~\eqref{EQUATION:SDE_WEALTH_WITH_BANKRUPTCY_TRANSFORMED} can be defined as
\begin{equation}
    \label{EQUATION:ADMISSIBLE_CONTROLS_ORIGINAL}
    \begin{split}
        \mathcal{A}_{w} := \big\{\pi \,|\, &(\pi_{t})_{t \in [0, T]} \text{ is adapted with c\'{a}dl\'{a}g paths and } \int_{0}^{T} \pi_{t}^{2} \mathrm{d}s < \infty \; \mathbb{P}\text{-a.s.}, \\
        &(W_{t})_{t \in [0, T]} \text{ is a unique strong solution such that } W_{0}^{\pi, \tau} = w, \\
        &W_{t}^{\pi, \tau}(\omega) > 0 \; \mathbb{P} \otimes \nu \text{ a.e.~in } [0, T] \times \Omega, \quad
        \mathbb{E} \big[\mathop{\operatorname{ess\,sup}}_{t \in [0, T]} U(W_{t}^{\pi, \tau})_{-}\big] < \infty\big\},
        \end{split}
\end{equation}
where $x_{-} = \max\{0, -x\}$ denotes the negative part of $x$ and $\nu$ is the Borel measure on $[0, \infty)$.
Thus, the optimal portfolio problem can be stated as
\begin{equation}
    \label{EQUATION:OPTIMAL_PORTFOLIO_PROBLEM}
    J(\pi) \to \max_{\pi \in \mathcal{A}_{w}}, \quad
    \pi^{\ast} = \mathop{\operatorname{arg\,max}}_{\pi \in \mathcal{A}_{w}} J(\pi), \quad W^{\ast} = (W_{t}^{\pi^{\ast}, \tau})_{t \in [0, T]},
\end{equation}
where $\pi^{\ast}$ is the optimal allocation process and $W^{\ast}$ is the associated optimal wealth process.

To solve the optimal control problem, we first want to explicitly characterize the forward map $\pi \mapsto W^{\pi, \tau}$ mapping each admissible control $\pi \in \mathcal{A}_{w}$ to the unique strong solution $(W^{\pi, \tau}_{t})_{t \geq 0}$ of Equation~\eqref{EQUATION:SDE_WEALTH_WITH_BANKRUPTCY}. To this end, let $(X_{t})_{t \geq 0}$ denote the unique strong solution to the classical wealth SDE without bankruptcy
\begin{equation}
    \label{EQUATION:SDE_WEALTH}
    \mathrm{d}X_{t} = \big(\mu \pi_{t} + r (1 - \pi_{t})\big) X_{t} \mathrm{d}t + \sigma \pi_{t} X_{t} \mathrm{d}B_{t} \quad \text{ for } t > 0, \quad
    X_{0} = w
\end{equation}
explicitly given~\cite[Chapter~4]{KlPl1992} via
\begin{equation}
    \label{EQUATION:X_T_PROCESS_EXPLICIT_SOLUTION}
    X_{t} =  \Phi_{t} w - \sigma \int_{0}^{t} \Phi_{t} \Phi_{s}^{-1} \pi_{s} \mathrm{d}s
\end{equation}
with
\begin{equation*}
    \Phi_{t} = \exp\Big(\int_{0}^{t} \Big(\big(\mu \pi_{s} + r (1 - \pi_{s})\big) - \frac{1}{2} \sigma^{2} \pi_{s}^{2}\Big) \mathrm{d}s + \sigma \int_{0}^{t} \pi_{s} \mathrm{d}B_{s}\Big).
\end{equation*}
With this notation, the wealth process $(W_{t})_{t \in [0, T]}$ solving Equation~\eqref{EQUATION:SDE_WEALTH_WITH_BANKRUPTCY} can be expressed as
\begin{equation}
    \label{EQUATION:ACTUAL_WEALTH}
    W_{t} = X_{t} \mathds{1}_{\{t < \tau\}} +
    \exp\Big(\int_{\tau}^{t} r(1 - \pi_{s}) \mathrm{d}s\Big) (1 - \pi_{\tau-}) X_{\tau} \mathds{1}_{\{t \geq \tau\}}.
\end{equation}
Indeed, if the bankruptcy has not happened at time $t$ or before, $W_{t}$ is same as $X_{t}$. Otherwise, all capital allocated in the stock at time $\tau-$ (viz.~$\pi_{\tau-} X_{\tau-}$) will be lost, while the remaining capital allocated in the riskless asset (viz.~$(1 - \pi_{\tau-}) X_{\tau-}$) will continue accruing from time $\tau$ to time $t$.

As can be seen from Equation~\eqref{EQUATION:ACTUAL_WEALTH}, the positivity of $W_{t} > 0$ $\mathbb{P} \otimes \nu$-a.e.~required in the definition of $\mathcal{A}_{w}$ in Equation~\eqref{EQUATION:ADMISSIBLE_CONTROLS_ORIGINAL} necessitates that
\begin{equation*}
    \pi_{t}(\omega) < 1 \text{ for } \mathbb{P} \otimes \nu \text{ for a.e. } (t, \omega) \in \Omega \times [0, T].
\end{equation*}
Indeed, letting $I = \{t \in [0, T] \,|\, \pi_{t} \geq 1\}$, we obtain
\begin{equation*}
    \begin{split}
        \mathbb{P}\{W_{t} \leq 0\} &= \mathbb{P}\Big\{\exp\Big(\int_{\tau}^{t} r(1 - \pi_{s}) \mathrm{d}s\Big) (1 - \pi_{\tau-}) X_{\tau} \mathds{1}_{\{T \geq \tau\}} \leq 0\Big\} \\
        &= \mathbb{P}\Big\{(1 - \pi_{\tau-}) X_{\tau} \mathds{1}_{\{\tau \in I\}} \leq 0\Big\} = \mathbb{P}\{\tau \in I\},
    \end{split}
\end{equation*}
where we used the fact that $X_{t} > 0$ $\mathbb{P}$-a.s.~for a.e.~$t \in [0, T]$. The latter probability vanishes if and only if $\nu(I) = 0$.

Thus, the admissible set in Equation~\eqref{EQUATION:ADMISSIBLE_CONTROLS_ORIGINAL} can be equivalently expressed as
\begin{equation}
    \label{EQUATION:ADMISSIBLE_CONTROLS_ORIGINAL}
    \begin{split}
        \mathcal{A}_{w} := \big\{\pi \,|\, &(\pi)_{t \in [0, T]} \text{ is adapted  with c\'{a}dl\'{a}g paths and } \int_{0}^{T} \pi_{t}^{2} \mathrm{d}t < \infty \; \mathbb{P}\text{-a.s.}, \\
        &\pi_{t}(\omega) < 1 \; \mathbb{P} \otimes \nu \text{ a.e.~in } \Omega \times [0, T], \\
        &(W_{t})_{t \in [0, T]} \text{ is a unique strong solution with } W_{0}^{\pi, \tau} = w, \\
        &\mathbb{E} \big[\mathop{\operatorname{ess\,sup}}_{t \in [0, T]} U(W_{t}^{\pi, \tau})_{-}\big] < \infty\big\}.
    \end{split}
\end{equation}

Analyzing Equation~\eqref{EQUATION:ACTUAL_WEALTH}, we easily observe that any wealth fraction $\pi_{t}$ invested into the stock after bankruptcy (i.e., $t \geq \tau$) in lieu of the riskless asset can only reduce the terminal wealth, $W_{T}$, and, thus, the expected utility $J(\cdot)$. Therefore, without loss of generality, we can assume $\pi_{t} = 0$ for $t \geq \tau$ leading to state-switching controls
\begin{equation}
    \label{EQUATION:STATE_DEPENDENT_CONTROL}
    \pi_{t} = \pi^{\ast}_{\mathrm{pre}} \mathds{1}_{\{t < \tau\}} + \pi^{\ast}_{\mathrm{post}} \mathds{1}_{\{t \geq \tau\}} =
    \left\{
    \begin{array}{ll}
        \pi_{\mathrm{pre}, t}, & t < \tau, \\
        \pi_{\mathrm{post}, t}, & t \geq \tau
    \end{array}
    \right. \quad \text{ with } \pi_{\mathrm{post}, t} \equiv 0.
\end{equation}
This furnishes a simplified version of Equation~\eqref{EQUATION:ACTUAL_WEALTH} given by
\begin{equation}
    \label{EQUATION:ACTUAL_WEALTH_REDUCED}
    W_{t} = X_{t} \mathds{1}_{\{t < \tau\}} + e^{r(t - \tau)} (1 - \pi_{\tau-}) X_{\tau} \mathds{1}_{\{t \geq \tau\}}
\end{equation}
without violating any admissibility conditions.

\section{Heuristic Derivation}
\label{SECTION:HEURISTIC_DERIVATION}

We start with a general isoelastic utility given in Equation~\eqref{EQUATION:ISOELASTIC_UTILITY_MERTON}. Define the set of admissible strategies
\begin{equation*}
    \begin{split}
        \mathcal{A}_{t, w} := \big\{\pi \,|\, &(\pi)_{s \in [t, T]} \text{ is adapted  with c\'{a}dl\'{a}g paths and } \int_{t}^{T} \pi_{s}^{2} \mathrm{d}s < \infty \\
        &\mathbb{P}\text{-a.s.}, \;
        \pi_{s}(\omega) < 1 \text{ for } \mathbb{P} \otimes \nu \text{ a.e.~in } \Omega \times [t, T], \\
        &(W_{s})_{s \in [t, T]} \text{ is a unique strong solution with } W_{t}^{\pi, \tau} = w, \\
        &\mathbb{E} \big[\mathop{\operatorname{ess\,sup}}_{s \in [t, T]} U(W_{s}^{\pi, \tau})_{-}\big] < \infty\big\},
    \end{split}
\end{equation*}
where $(W_{s}^{\pi, \tau})_{s \in [t, T]}$ solves the jump SDE in Equation~\eqref{EQUATION:SDE_WEALTH_WITH_BANKRUPTCY} over $[t, T]$. Further, consider the value function
\begin{equation}
     \label{EQUATION:VALUE_FUNCTION_HEURISTIC}
     V(t, w) = \sup_{\pi \in \mathcal{A}_{t, w}} \mathbb{E}\big[U(W^{\pi, \tau}_{T}) \,|\, W_{t}^{\pi, \tau} = w\big].
\end{equation}
In spirit of~\cite{AzPiWe2014}, we also introduce two auxiliary `conditional' value functions
\begin{equation}
    \begin{split}
        V^{\mathrm{pre}}(t, w) &= \sup_{\pi \in \mathcal{A}_{t, w}} \mathbb{E}\big[U(W^{\pi, \tau}_{T}) \,|\, W_{t}^{\pi, \tau} = w, t < \tau\big], \\
        V^{\mathrm{post}}(t, w) &= \sup_{\pi \in \mathcal{A}_{t, w}} \mathbb{E}\big[U(W^{\pi, \tau}_{T}) \,|\, W_{t}^{\pi, \tau} = w, t \geq \tau\big] = \mathbb{E}\big[U(W^{\pi = 0, \tau}_{T}) \,|\, W_{t}^{\pi, \tau} = w, t \geq \tau\big]
    \end{split}
\end{equation}
corresponding to the case no bankruptcy has occurred at or before time $t$ and the opposite case, respectively. It should be noted that this framework allows for default-specific optimal control processes, i.e., $\pi_{t} = \pi_{t}(\mathds{1}_{\{t < \tau\}})$.

Using the law of total probability, the total value function in Equation~\eqref{EQUATION:VALUE_FUNCTION_HEURISTIC} can be expressed as
\begin{equation}
    \label{EQUATION:VALUE_FUNCTION_HEURISTIC_IDENTITY}
    \begin{split}
        V(t, w) &= V^{\mathrm{pre}}(t, w) \, \mathbb{P}\{\tau < t\} + V^{\mathrm{post}}(t, w) \, \mathbb{P}\{\tau \geq t\} \\
        &= (1 - e^{-\lambda t}) V^{\mathrm{pre}}(t, w) + e^{-\lambda t} V^{\mathrm{post}}(t, w).
    \end{split}
\end{equation}

\subsection{Logarithmic Utility}
\label{SECTION:LOGARITHMIC_UTILITY}

Choosing $\gamma = 1$, Equation~\eqref{EQUATION:ISOELASTIC_UTILITY_MERTON} reduces to the logarithmic utility function
\begin{equation}
    \label{EQUATION:LOGARITHMIC_UTILITY}
    U(W) = \log(W).
\end{equation}
Arguing similar to Equation~\eqref{EQUATION:ACTUAL_WEALTH_REDUCED}, we can explicitly compute the `post' conditional value function
\begin{equation*}
    V^{\mathrm{post}}(t, w) = U\big(e^{r(T - t)} (1 - \pi_{t-}) w\big)
    = r(T - t) + \log(1 - \pi_{t-}) + \log(w).
\end{equation*}

As for the `pre' conditional value function, given the time of crash is distributed as the first arrival time of a homogeneous Poisson process, the conditional probability of no crash occurring in $[t, t + \mathrm{d}t]$ assuming no crash has occurred until time $t$ is $(1 - \lambda \mathrm{d}t) + o(\mathrm{d}t)$, while the conditional probability of a crash is $\lambda \mathrm{d}t + o(\mathrm{d}t)$. Thus, using an argumentation similar to \cite[Theorem~3.3]{AzPiWe2014}, although without having pathwise continuity of the controlled wealth process in our situation, we can heuristically derive an infinitesimal recursive formula for $V^{\mathrm{pre}}(t, w)$ reading as
\begin{equation*}
    \begin{split}
        \mathrm{d} V^{\mathrm{pre}}(t, w) &= (1-\lambda \mathrm{d}t) \mathrm{d} V^{\mathrm{pre}}(t, w)
        + \lambda \mathrm{d}t \big(V^{\mathrm{post}}(t, w) - V^{\mathrm{pre}}(t, w)\big) + o(\mathrm{d}t) \\
        &= (1-\lambda \mathrm{d}t) \mathrm{d} V^{\mathrm{pre}}(t, w)  \\
        &+ \lambda \mathrm{d}t \big(r(T - t) + \log(1 - \pi_{t-}) + \log(w) - V^{\mathrm{pre}}(t, w)\big) + o(\mathrm{d}t).
    \end{split}
\end{equation*}

Let $\pi_{t} \in \mathcal{A}_{t, w}$. For a wealth process $(W_{s}^{\pi, \tau})_{s \in [t, T]}$ solving the jump SDE in Equation~\eqref{EQUATION:SDE_WEALTH_WITH_BANKRUPTCY} subject to the initial condition $W_{t}^{\pi, \tau} = w$, we evaluate the differential $\mathrm{d}V^{\mathrm{pre}}(s, W_{s})$ for $t \leq s \leq T$. According to Equation~\eqref{EQUATION:SDE_WEALTH}, the wealth process satisfies
\begin{equation*}
    \begin{split}
         \mathrm{d}W_{s} &= W_{s}\left(\pi_{s} \mu \mathrm{d}s + \pi_{s} \sigma \mathrm{d}B_{s} + (1 - \pi_{s}) r \mathrm{d}s\right), \\
         (\mathrm{d}W_{s})^{2} &= W_{s}^{2} \pi_{s}^2 \sigma^{2} \mathrm{d}s.
    \end{split}
\end{equation*}
Assuming $V^{\mathrm{pre}}(\cdot)$ is a smooth function, It\^{o}'s rule (cf.~\cite[Theorem~3.5]{AzPiWe2014}) furnishes
\begin{equation}
    \label{ItoJump}
    \begin{split}
        \mathrm{d} V^{\mathrm{pre}}(t, w) &= \partial_{w} V^{\mathrm{pre}}(t, w) (w \left(\pi_{t} \mu \mathrm{d}t + (1 - \pi_{t}) r \mathrm{d}t\right) \\
        &+ \frac{1}{2} \partial_{w}^{2} V^{\mathrm{pre}}(t, w) w^{2} \pi_{t}^2 \sigma^{2} \mathrm{d}t
        + \partial_{t} V^{\mathrm{pre}}(t, w) \mathrm{d}t \\
        &+ \lambda \mathrm{d}t \left(r(T - t) + \log(1 - \pi_{t}) + \log(w) - V^{\mathrm{pre}}(t, w)\right),
    \end{split}
\end{equation}
where we assume the continuity of $\pi_{s}$ (and, therefore, that of $\pi_{\mathrm{pre}, s}$) for $t \leq s < \tau$ implying $\pi_{s-} = \pi_{s}$. Dividing by $\mathrm{d}t$ and passing to the limit $\mathrm{d}t \to 0$, we obtain an HJB-type equation for the `pre' conditional value function
\begin{equation}
    \label{HJMEq}
    \partial_{t} V^{\mathrm{pre}} + \mathcal{H}(t, V^{\mathrm{pre}}, \partial_{w} V^{\mathrm{pre}}, \partial_{w}^{2} V^{\mathrm{pre}}) = 0
\end{equation}
where
\begin{equation}
    \label{EQUATION:HJB_OPERATOR_LOGARITHMIC_UTILITY}
    \begin{split}
        \mathcal{H}(t, w, V^{\mathrm{pre}}, \partial_{w} V^{\mathrm{pre}}, &\partial_{ww} V^{\mathrm{pre}}) = \sup_{\pi_{t} < 1} \Big(\partial_{w} V^{\mathrm{pre}} (w \left(\pi_{t} \mu + (1 - \pi_{t}) r\right) \\
        &+ \tfrac{1}{2} \partial_{w}^{2} V^{\mathrm{pre}} w^{2} \pi_{t}^2 \sigma^{2}
        + \lambda \left(r(T - t) + \log(1 - \pi_{t}) + \log(w) - V^{\mathrm{pre}}\right)\Big).
    \end{split}
\end{equation}
In contrast to usual HJB equation, we can observe that our HJB operator $\mathcal{H}(\cdot)$ additionally depends on the derivatives of the value function but also the value function itself.

Using the ansatz
\begin{equation}
    \label{EQUATION:ANSATZ_SOLUTION_SWITCHING_ODE}
    V^{\mathrm{pre}}(t, w) = \log(w) + f(t),
\end{equation}
the first-order Fermat optimality condition for the function under the supremum reads as
\begin{equation}
    (\mu - r) - \pi_{t} \sigma^{2} - \frac{\lambda}{1 - \pi_{t}} = 0.
\end{equation}
Arguing similar to Appendix~\ref{APPENDIX:AUXILIARY_RESULTS}, the unique maximum is attained at
\begin{equation*}
    \pi_{\mathrm{pre}}^{\ast} = \frac{\mu - r + \sigma^2 - \sqrt{\left(\sigma^2 - (\mu - r)\right)^2 + 4\lambda\sigma^2}}{2\sigma^2}.
\end{equation*}

Plugging $\pi_{t} \equiv \pi^{\ast}_{\mathrm{pre}}$ back into Equation~\eqref{HJMEq}, we obtain an ordinary differential equation (ODE) for $f(t)$
\begin{equation}
    \label{EQUATION:PRE}
    f'(t) + \lambda \big(r(T-t)+\log(1-\pi_t)-f(t) \big)
     -\tfrac{1}{2}\pi_t^2\sigma^2 + \pi_t(\mu-r) + r =0
\end{equation}
with the terminal condition $f(T) = 0$ we can explicitly solve.
Introducing the constant
\begin{equation*}
    C= -\tfrac{1}{2} (\pi_{\mathrm{pre}}^{\ast})^{2} \sigma^2
    + \pi_{\mathrm{pre}}^{\ast}(\mu - r) + r + \lambda \log(1-\pi_{\mathrm{pre}}^{\ast}),
\end{equation*}
we rewrite Equation \eqref{EQUATION:PRE} as
\begin{equation*}
    f'(t) - \lambda f(t) + \lambda r(T-t) + C = 0
\end{equation*}
with the general solution given by
\begin{equation*}
    f(t) = B\mathrm{e}^{\lambda t} + r(T-t)+ \frac{C-r}{\lambda} \quad \text{ for some constant } B.
\end{equation*}
The terminal condition $f(T)=0$ furnishes $B = \mathrm{e}^{-\lambda T} (r - C)$ and the `pre' value function is expressed as
\begin{equation*}
    V^{\mathrm{pre}}(t, w) = \log(w) + r(T-t) + \frac{C - r}{\lambda}
    \big(1-\mathrm{e}^{-\lambda (T-t)}\big).
\end{equation*}

\subsection{Power Utility ($\gamma \neq 1$)}

For $\gamma \neq 1$, without loss of generality, the isoelastic utility in Equation~\eqref{EQUATION:ISOELASTIC_UTILITY_MERTON} can be reduced to the power utility
\begin{equation}
    \label{EQUATION:POWER_UTILITY}
    U(W) = \frac{W^{1-\gamma}}{1-\gamma}
\end{equation}
by dropping the constant term $-\frac{1}{1-\gamma}$. Similar to Section~\ref{SECTION:LOGARITHMIC_UTILITY}, the `post' conditional value function can be expressed as
\begin{equation}
    V^{\mathrm{post}}(t, w) = e^{(1-\gamma) r(T-t)} \frac{w^{1-\gamma}}{1-\gamma} (1 - \pi_{t})^{1-\gamma}.
\end{equation}
Similarly, invoking Ito's lemma, we can write
\begin{align}
    \label{ItoJumpUtility}
    \mathrm{d} V^{\mathrm{pre}}(t, w) &=
    \Big(\partial_{t} V^{\mathrm{pre}}(t, w)
    + \mathcal{H}(t, w, \pi_t, V^{\mathrm{pre}}, \partial_{w} V^{\mathrm{pre}}, \partial_{ww} V^{\mathrm{pre}})
    \Big) \mathrm{d}t \equiv 0
\end{align}
with the HJB operator
\begin{equation}
    \label{HItoJump}
    \begin{split}
        \mathcal{H}(t, w, &V^{\mathrm{pre}}, \partial_{w} V^{\mathrm{pre}}, \partial_{ww} V^{\mathrm{pre}})
        \\
        & = \sup_{\pi_{t} < 1} \Big\{w \big(\mu \pi_{t} + r (1 - \pi_{t}) \big) \partial_{w} V^{\mathrm{pre}}(t, w)
        + \frac{1}{2} \sigma^{2} \pi_{t}^2 w^{2} \partial_{w}^{2} V^{\mathrm{pre}}(t, w) + \\
        & \phantom{\sup_{\pi_{t} < 1} \Big\{} \lambda \big(V^{\mathrm{post}}(t, w) - V^{\mathrm{pre}}(t, w)\big)\Big\} \\
        &= \sup_{\pi_{t} < 1} \Big\{w \big(\mu \pi_{t} + r (1 - \pi_{t}) \big) \partial_{w} V^{\mathrm{pre}}(t, w)
        + \frac{1}{2} \sigma^{2} \pi_{t}^2 w^{2} \partial_{w}^{2} V^{\mathrm{pre}}(t, w) + \\
        & \phantom{\sup_{\pi_{t} < 1} \Big\{} \lambda \big(e^{(1-\gamma) r(T-t)} \frac{w^{1-\gamma}}{1-\gamma} (1 - \pi_{t})^{1-\gamma} - V^{\mathrm{pre}}(t, w)\big)\Big\}.
    \end{split}
\end{equation}
This leads to the HJB equation
\begin{equation}
    \label{HJMEq.POWER_UTILITY}
    \partial_{t} V^{\mathrm{pre}} + \mathcal{H}(t, V^{\mathrm{pre}}, \partial_{w} V^{\mathrm{pre}}, \partial_{w}^{2} V^{\mathrm{pre}}) = 0.
\end{equation}

Plugging the ansatz
\begin{equation}
    \label{EQUATION:ANSATZ_POWER_UTILITY}
    V^{\mathrm{pre}}(t, w) = f(t) \frac{w^{1-\gamma}}{1 - \gamma} e^{(1-\gamma) r(T-t)}
\end{equation}
into Equation~\eqref{HJMEq.POWER_UTILITY} and dividing both sides by $e^{(1-\gamma)(T-t)r} w^{1-\gamma}$, we obtain
\begin{equation}
    \label{HJMEq.POWER_UTILITY_ANSATZ_PLUGGED}
    \frac{1}{1 - \gamma} e^{-(1-\gamma)(T-t)r} \partial_{t} \big(f(t) e^{(1-\gamma) r(T-t)}\big) +
    \sup_{\pi_{t} < 1} h(t, \pi_{t}) = 0
\end{equation}
with
\begin{equation}
    \label{EQUATION:HJB_POWER_UTILITY_H_FUNCTION_REDUCED}
    \begin{split}
        h(t, \pi) &= -\frac{1}{\gamma-1}
        \Big(\lambda (1-\pi)^{1-\gamma} +
        \big(\frac{\gamma}{2} (\gamma-1) \sigma^2 \pi^2
        - (\gamma-1) (\mu-r) \pi - \lambda\big) f(t) \Big).
    \end{split}
\end{equation}

The Fermat's optimality condition for the maximization problem in Equation~\eqref{HJMEq.POWER_UTILITY_ANSATZ_PLUGGED} reads as
\begin{equation}
    \label{OptCondGC}
    \left(\mu-r - \gamma \sigma^2 \pi_t \right) f(t) - \lambda(1-\pi_t)^{-\gamma} = 0
\end{equation}
implying that $\pi_t$ and $f(t)$ are algebraically related.
Likewise, plugging the ansatz into PDE~\eqref{HJMEq.POWER_UTILITY_ANSATZ_PLUGGED} and canceling out $e^{(1 - \gamma) r(T - t)}$, we arrive at the ODE
\begin{equation}
    \label{ODEfGC}
    \frac{f'(t)}{\gamma-1} +\left(\frac{\gamma}{2}\sigma^2\pi_t^2
    -(\mu-r)\pi_t - \frac{\lambda}{\gamma-1} \right)f(t) + \frac{\lambda}{\gamma-1}(1-\pi_t)^{1-\gamma} =0.
\end{equation}

Solving Equation~\eqref{OptCondGC} for $f(t)$, we obtain
\begin{equation}
    \label{EQUATION:F_VIA_PI}
    f(t) = -\frac{\lambda(1-\pi_t)^{-\gamma}}{\gamma\sigma^2\pi_t- (\mu-r)}.
\end{equation}
Plugging this into Equation~\eqref{EQUATION:HJB_POWER_UTILITY_H_FUNCTION_REDUCED} eliminates $f(t)$ and yields the reduced form of $h$:
\begin{equation}
    \label{EQUATION:HJB_POWER_UTILITY_H_FUNCTION_REDUCED_F_ELIMINATED}
    \begin{split}
        h(t, \pi) &= -\frac{\lambda (1-\pi)^{1-\gamma}}{\gamma-1}
        \Big(1  - \frac{\frac{\gamma}{2} (\gamma-1) \sigma^2 \pi^2
        - (\gamma-1) (\mu-r) \pi - \lambda}{\gamma\sigma^2\pi - (\mu-r)} \Big).
    \end{split}
\end{equation}
Likewise, plugging the representation in Equation~\eqref{EQUATION:F_VIA_PI} into Equation~\eqref{ODEfGC} furnishes the ODE
\begin{equation}
    \label{EQUATION:ODE_FOR_PI_T}
    \frac{\mathrm{d} \pi_t}{\mathrm{d} t} = \kappa(\pi_{t}) \quad \text{ with } \quad
    \kappa(\pi) =
    \frac{\gamma \sigma^2
    (1-\pi) \big(\alpha - \pi\big)
    \big(\tfrac{1}{2} \pi^2 - \beta \pi + \eta\big)} {\pi  - \beta}
\end{equation}
with the constants
\begin{equation}
    \label{EQUATION:CONSTANTS_ALPHA_BETA_ETA}
    \alpha = \frac{\mu-r}{\gamma \sigma^2}, \quad
    \beta = \frac{\mu -r  + \sigma^2}{(\gamma+1)\sigma^2}, \quad
    \eta = \frac{\mu - r - \lambda}{\gamma (\gamma+1) \sigma^2},
\end{equation}
where $\alpha$ is the traditional Merton ratio.
Equation~\eqref{EQUATION:ODE_FOR_PI_T} is an ODE in $\pi_{t}$ instead of $f(t)$ as the conventional approach would suggest. This will allow us to solve the equation analytically.

Recalling the terminal condition $f(T) = 1$, Equation~\eqref{OptCondGC} yields the algebraic equation
\begin{equation}
    \label{piTEq}
    \phi(\pi_{T}) = 0 \quad \text{ with } \quad
    \phi(\pi_T) := \gamma\sigma^2 \big(\alpha - \pi_T\big) - \lambda(1-\pi_T)^{-\gamma}
\end{equation}
that can be solved for $\pi_{T}$. In the limiting case $\gamma = 1$, it is interesting to observe that $\phi(\pi_{T}) = 0$ reduces to $\tfrac{1}{2} \pi_T^2 - \beta \pi_T + \eta = 0$ implying $\kappa(\pi_{T}) = 0$ and, thus, $\pi_t \equiv \pi_T$ is constant consistent with Section~\ref{SECTION:LOGARITHMIC_UTILITY}.

By the virtue of Proposition~\ref{PROPOSITION:ODE_SOLUTION_PLUS_MAXIMIZER} in Appendix Section~\ref{APPENDIX:AUXILIARY_RESULTS_GENERAL_GAMMA}, the ODE in Equation~\eqref{EQUATION:ODE_FOR_PI_T} subject to the terminal condition $\phi(\pi_{T}) = 0$ is uniquely solvable such that $\pi_{t}$ exists for all $t \in [0, T]$ and satisfies $\pi_{t} = \max_{\pi < 1} h(t, \pi)$ as Equation~\eqref{HJMEq.POWER_UTILITY_ANSATZ_PLUGGED} mandates.
Moreover, $\pi_{t}$ never exceeds the classical Merton ratio or $1$, i.e.,
\begin{equation*}
    \pi_{t} \leq \frac{\mu - r}{\gamma \sigma^{2}} \quad \text{ and } \quad \pi_{t} < 1 \quad \text{ for all } t \in [0, T],
\end{equation*}
which is intuitively understandable since a more conservative strategy is naturally expected given the risk of bankruptcy.

To further solve the ODE in Equation~\eqref{EQUATION:ODE_FOR_PI_T} explicitly, we observe it is autonomous and, therefore, separable. Separating the variables and integrating, the solution $\pi_{t}$ is given implicitly via
\begin{equation}
    \Psi(\pi_t) = \Psi(\pi_T) e^{-(T-t)}
\end{equation}
with
\begin{align}
    \Psi(\pi_t) := (1-\pi_t)^{2(1-\beta) d_1}(\alpha - \pi_t)^{-2(\beta-\alpha)d_2}
    (\pi_t-q_-)^{D-(d_1-d_2)\sqrt{\Delta}}(q_+ -\pi_t)^{D+(d_1-d_2)\sqrt{\Delta}},
\end{align}
where
\begin{align*}
    d_1 &= \frac{1}{\gamma \sigma^2 (1-\alpha)(2\beta-2\eta-1)},\quad
    d_2 = \frac{1}{\gamma \sigma^2 (1-\alpha) (2\alpha \beta -\alpha^2 -2\eta)},\\
    D &= (1-\beta) d_1 + (\beta-\alpha) d_2, \quad \Delta = \beta^2 - 2\eta, \quad
    q_- = \beta - \sqrt{\Delta}, \quad q_+ = \beta + \sqrt{\Delta}.
\end{align*}
Backsubstituting $\pi_{t}$ into Equations~\eqref{EQUATION:F_VIA_PI} and \eqref{EQUATION:ANSATZ_POWER_UTILITY}, we obtain the `pre' value function $V^{\mathrm{pre}}(t, w)$  that solves the HJB in Equation~\eqref{HJMEq.POWER_UTILITY}.

The verification theorem for Markov-switching SDEs established by~\cite{AzPiWe2014} does not immediately apply to Equation~\eqref{HJMEq} or Equation~\eqref{HJMEq.POWER_UTILITY} due to the lack of pathwise continuity of solution processes to Equation~\eqref{EQUATION:SDE_WEALTH_WITH_BANKRUPTCY} as well as discontinuity of the utility function at zero. While we expect the results of~\cite{AzPiWe2014} can potentially be extended to our situation, we pursue a simpler strategy in Section~\ref{SECTION:RIGOROUS_DERIVATION} by transforming the problem to a conventional controlled Markovian diffusion, for which rigorous verification results exist in the literature.

\section{Reduction to Controlled Markovian Diffusion}
\label{SECTION:RIGOROUS_DERIVATION}

In this section, we present an alternative derivation that reduces the non-conventional optimal control problem discussed in Section~\ref{SECTION:PRELIMINARIES} to the classical case. For illustration purposes, we will solely focus on the logarithmic utility.
For the sake of simplicity, with some abuse of notation, we will write $\pi_{t}$ to denote $\pi_{\mathrm{pre}, t}$ for the most part of this section until Equation~\eqref{EQUATION:MERTON_RATIO}, where $\pi_{t}$ reclaims its original meaning (cf.~Equation~\eqref{EQUATION:STATE_DEPENDENT_CONTROL}).

Our approach is inspired by~\cite{Kli1970}. Assuming the conditions of Fubini \& Tonelli's theorem are satisfied, we can write
\begin{align*}
    \mathbb{E}_{\mathbb{P}_{W} \otimes \mathbb{P}_{\tau}}&\big[u(Y_{T}^{\pi})\big] =
    \mathbb{E}_{\mathbb{P}_{W}} \Big[\mathbb{E}_{\tau} \big[u(Y_{T}^{\pi})\big]\Big]
    = \mathbb{E}_{\mathbb{P}_{W}}\Big[\int_{0}^{\infty} u(Y_{T}^{\pi}) \varphi(t) \mathrm{d}t\Big] \\
    &= \mathbb{E}_{\mathbb{P}_{W}}\Big[\int_{0}^{\infty} \Big(u(X_{T}) \mathds{1}_{\{T < t\}} + u\big(e^{r(T - t)} (1 - \pi_{t}) X_{t}\big) \mathds{1}_{\{T \geq t\}}\Big) \varphi(t) \mathrm{d}t\Big] \\
    &= \mathbb{E}_{\mathbb{P}_{W}} \Big[\int_{T}^{\infty} u(X_{T}) \varphi(t) \mathrm{d}t\Big]
    + \mathbb{E}_{\mathbb{P}_{W}}\Big[\int_{0}^{T} u\big(e^{r(T - t)} (1 - \pi_{t}) X_{t}\big) \varphi(t) \mathrm{d}t\Big] \\
    &= \mathbb{E}_{\mathbb{P}_{W}} \Big[\Big(\int_{T}^{\infty} \varphi(t) \mathrm{d}t\Big) u(X_{T})\Big]
    + \mathbb{E}_{\mathbb{P}_{W}}\Big[\int_{0}^{T} u\big(e^{r(T - t)} (1 - \pi_{t}) X_{t}\big) \varphi(t) \mathrm{d}t\Big] \\
    &=
    \mathbb{P}_{\mathbb{P}_{\tau}}\{\tau \geq T\} \mathbb{E}_{\mathbb{P}_{W}} \big[u(X_{T})\big]
    + \mathbb{E}_{\mathbb{P}_{W}}\Big[\int_{0}^{T} u\big(e^{r(T - t)} (1 - \pi_{t}) X_{t}\big) \varphi(t) \mathrm{d}t\Big]
\end{align*}
where $\mathbb{P}_{\mathbb{P}_{\tau}}\{\tau \geq T\} = \int_{T}^{\infty} \varphi(t) \mathrm{d}t$ solely depends on $T$ with $\varphi(\cdot)$ denoting the probability density function of $\tau$. As before, we choose $u(x) = \log(x)$ to be the logarithmic utility and assume $\tau \sim \mathrm{Exp}(\lambda)$ is exponentially distributed with some known intensity $\lambda > 0$, i.e., $\varphi(t) = \lambda \exp(-\lambda t) \mathds{1}_{\{t \geq 0\}}$. Then the expected utility can be further simplified as
\begin{align}
    \label{EQUATION:STOCHASTIC_CONTROL_PROBLEM_REDUCED}
    \begin{split}
        \mathbb{E}_{\mathbb{P}_{W} \otimes \mathbb{P}_{\tau}}\big[u(Y_{T}^{\pi})\big] &=
        \mathbb{E}_{\mathbb{P}_{W}} \Big[\int_{0}^{T} \lambda e^{-\lambda t} \log\big((1 - \pi_{t}) X_{t}\big) \mathrm{d}t\Big]
         \\
        &+ \mathbb{E}_{\mathbb{P}_{W}} \big[e^{-\lambda T} \log(X_{T})\big]
        + r \lambda \int_{0}^{T} (T - t) e^{-\lambda t} \mathrm{d}t \\
        &=
        \mathbb{E}_{\mathbb{P}_{W}} \Big[\int_{0}^{T} \lambda e^{-\lambda t} \log\big((1 - \pi_{t}) X_{t}\big) \mathrm{d}t\Big] \\
        &+ \mathbb{E}_{\mathbb{P}_{W}} \Big[e^{-\lambda T} \log(X_{T}) + r\Big(T - \frac{1 - e^{-\lambda T}}{\lambda}\Big)\Big]
    \end{split}
\end{align}
where $\tau$ is completely eliminated from the latter equation.

Therefore, we were able to reduce the original control problem with sporadic bankruptcy in Equations~\eqref{EQUATION:SDE_WEALTH_WITH_BANKRUPTCY}, \eqref{EQUATION:LOGARITHMIC_UTILITY}, \eqref{EQUATION:EXPECTED_TERMINAL_UTILITY} and \eqref{EQUATION:OPTIMAL_PORTFOLIO_PROBLEM} to a classical stochastic control problem in Equations~\eqref{EQUATION:SDE_WEALTH} and \eqref{EQUATION:STOCHASTIC_CONTROL_PROBLEM_REDUCED} with running and terminal costs for wealth $W_{t}$ so that the usual Bellman's principle for Markovian diffusions can now be applied. Throughout the rest of the section, we adopt the notation of~\cite[Chapter~IV]{FleSo2006}. Note that maximization of expected utility is equivalently expressed as minimimization of negative utility within this framework. Thus, we seek to minimize the functional
\begin{equation}
    \label{EQUATION:OBJECTIVE_FUNCTIONAL_CONTROLLED_DIFFUSION}
    \begin{split}
        J(\pi) &= \mathbb{E}\Big[-\int_{0}^{T} \lambda e^{-\lambda t} \log\big((1 - \pi_{t}) X_{t}\big) \mathrm{d}t\Big] \\
        &+ \mathbb{E}\Big[-e^{-\lambda T} \log(X_{T}) - r\Big(T - \frac{1 - e^{-\lambda T}}{\lambda}\Big)\Big].
    \end{split}
\end{equation}
Since the terminal cost function in Equation~\eqref{EQUATION:OBJECTIVE_FUNCTIONAL_CONTROLLED_DIFFUSION} cannot be continuously extended to $x = 0$, classical results (viz.~\cite{FleRi1975, FleSo2006}) are not directly applicable to our problem. Instead, we will employ the verification theorem proved in the recent work of~\cite{AlPe2021}. Other known results, e.g., \cite{FoSiZa2017, GoKa2000, PuWo2015}, do not appear to be directly applicable in our situation.

The set of admissible controls is defined as
\begin{equation}
    \label{EQUATION:ADMISSIBLE_CONTROLS_X_PROCESS}
    \begin{split}
        \mathcal{A}_{t, x} := \big\{\pi \,|\, &(\pi)_{s \in [t, T]} \text{ is adapted  with } \int_{t}^{T} \pi_{s}^{2} \mathrm{d}s < \infty \; \mathbb{P}\text{-a.s.}, \\
        &(X_{s})_{s \in [t, T]} \text{ is a unique strong solution to SDE~\eqref{EQUATION:SDE_WEALTH}} , \\
        &\text{with } X_{t}^{\pi} = x \text{ and } \mathbb{E} \big[\mathop{\operatorname{ess\,sup}}_{s \in [t, T]} \big(\log(X_{s}^{\pi})\big)_{-}\big] < \infty\big\}
    \end{split}
\end{equation}
where, unlike Equation~\eqref{EQUATION:SIGMA_ALGEBRA_B_T_AND_XI_T}, the filtration is now solely generated by the Brownian motion $(B_{t})_{t}$. Arguing similar to previous sections, the last condition in the definition of $\mathcal{A}_{t, x}$ in Equation~\eqref{EQUATION:ADMISSIBLE_CONTROLS_X_PROCESS} necessitates
\begin{equation*}
    \pi_{t}(\omega) < 1 \text{ for $(\mathbb{P} \otimes \nu)$-a.e. } (\omega, t) \in \Omega \times [0, T].
\end{equation*}

Introducing the Hamilton-Jacobi-Bellman (HJB) operator
\begin{align*}
    (t, V, \pi) \mapsto \mathcal{H}(t, \pi, \partial_{x} V, \partial_{xx} V)
\end{align*}
with
\begin{align}
    \label{EQUATION:BELLMAN_OPERATOR}
    \begin{split}
        \mathcal{H}(t, x, p, A) &=
        \inf_{\pi \in \mathbb{R}} \Big(f(t, x, \pi) \cdot p + \frac{1}{2} \sigma^{2}(t, x, \pi) \cdot A - L(t, x, \pi)\Big) \\
        &= \inf_{\pi < 1} \Big(f(t, x, \pi) \cdot p + \frac{1}{2} \sigma^{2}(t, x, \pi) \cdot A - L(t, x, \pi)\Big)
    \end{split}
\end{align}
where $f(t, x, \pi) = \big(\mu \pi + r (1 - \pi)\big) x$,
$\sigma(t, x, \pi) = \sigma \pi x$,
$L(t, x, \pi) = \lambda e^{-\lambda t} \log\big((1 - \pi) x\big)$,
consider the terminal value problem for the HJB partial differential equation (PDE)
\begin{align}
    \label{EQUATION:HJB_PDE}
    \partial_{t} V(t, x) + \mathcal{H}\big(t, x, \partial_{x} V(t, x), \partial_{xx} V(t, x)\big) &= 0 \text{ for } (t, x) \in (0, T) \times (0, \infty), \\
    \label{EQUATION:HJB_TC}
    V(T, x) &= \Psi(T, x) \text{ for } x \in (0, \infty)
\end{align}
with the terminal value $\Psi(T, x) = -e^{-\lambda T} \log(x) - r\Big(T - \frac{1 - e^{-\lambda T}}{\lambda}\Big)$. On the strength of Lemma~\ref{LEMMA:HJB_OPERATOR_MINIMIZATION} in the Appendix, the HJB operator can be expressed as a fully nonlinear elliptic operator
\begin{equation}
    \mathcal{H}(t, x, p, A) = \big(\mu \pi^{\ast} + r (1 - \pi)\big) x p + \frac{1}{2} \sigma^{2} (\pi^{\ast})^{2} x^{2} - \lambda e^{-\lambda t} \log\big((1 - \pi^{\ast}) x\big) \text{ for } A > 0
\end{equation}
with $\pi^{\ast} = \pi^{\ast}(t, x, p, A)$ given in Equation~\eqref{EQUATION:OPTIMAL_PROPORTION}. For $A \leq 0$, the HJB operator is undefined since the underlying minimization problem is ill-posed.

To solve the HJB Equations~\eqref{EQUATION:HJB_PDE}--\eqref{EQUATION:HJB_TC}, we use the the ansatz
\begin{equation}
        \label{EQUATION:VALUE_FUNCTION_ANSATZ}
        V(t, x) = -e^{-\lambda t} \log(x) + G(t)
\end{equation}
with $\partial_{xx} V(t, x) = e^{-\lambda t} (1/x^{2}) > 0$ for $t \geq 0$ and $x > 0$. Plugging into Equation~\eqref{EQUATION:OPTIMAL_PROPORTION} yields the (constant) pre-default Merton ratio:
\begin{equation}
    \label{EQUATION:MERTON_RATIO}
    \pi^{\ast}_{\mathrm{pre}} \equiv \frac{(\mu - r) + \sigma^{2} - \sqrt{(\sigma^{2} - (\mu - r))^{2} + 4 \lambda \sigma^{2}}}{2 \sigma^{2}}.
\end{equation}

Like the classical Merton ratio, the Merton ratio in Equation~\eqref{EQUATION:MERTON_RATIO} neither depends on $x$, nor $t$. See~Figure~\ref{FIGURE:MERTON_RATIO}.
\begin{figure}[h!]
    \centering
    \includegraphics[width=0.65\textwidth]{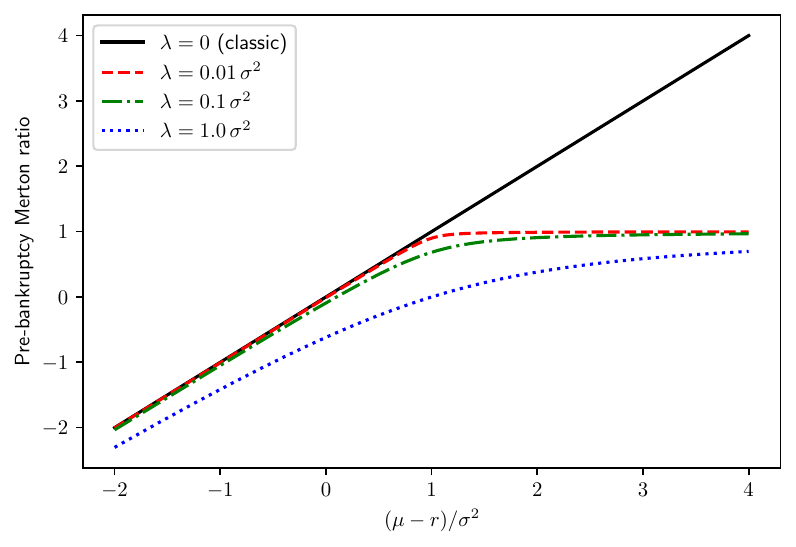}

    \caption{\label{FIGURE:MERTON_RATIO}
    Pre-default Merton ratio $\pi^{\ast}_{\mathrm{pre}}$ vs.~$\lambda$}
\end{figure}
Moreover, for $\mu - r \leq \sigma^{2}$, $\pi^{\ast}_{\mathrm{pre}}(\lambda)$ converges to the classical expression $\frac{\mu - r}{\sigma^{2}}$ as $\lambda \searrow 0$.  It is also interesting to observe that our Merton ratio in Equation~\eqref{EQUATION:MERTON_RATIO} always suggests to put less money into the stock than the usual Merton ratio would recommend. In particular, it also encourages to short the stock more aggressively, especially for larger $\lambda$, which makes practical sense in view of the chance of default. For $\mu - r < \sigma^{2}$, $\pi^{\ast}_{\mathrm{pre}}(\lambda) \nearrow 1$ as $\lambda \searrow 0$. Unlike classical situation, this means the investor should never invest all of his capital into stock or borrow from the bank to purchase stock. Indeed, since the risk of the stock price to default sporadically over any non-trivial time interval is positive, a Merton ratio $\pi \geq 1$ is not admissible in the sense of Equation~\eqref{EQUATION:ADMISSIBLE_CONTROLS_X_PROCESS} as it renders the terminal wealth $Y_{T}$ in Equation~\eqref{EQUATION:ACTUAL_WEALTH} zero or negative with a positive probability, therefore, rendering the expected utility equal $-\infty$. This is a major qualitative difference to the classical scenario.

\begin{proposition}
    The value function is given by
    \begin{equation}
        V(t, x) = -e^{-\lambda t} \log(x) -\frac{C^{\ast}(\lambda)}{\lambda} \big(e^{-\lambda t} - e^{-\lambda T}\big) - r\Big(T - \frac{1 - e^{-\lambda T}}{\lambda}\Big)
    \end{equation}
    for $t \geq 0, x > 0$ with $\pi^{\ast}_{\mathrm{pre}}$ from Equation~\eqref{EQUATION:MERTON_RATIO} and
    \begin{equation*}
        C^{\ast}(\lambda) =
        \big(\mu \pi^{\ast}_{\mathrm{pre}} + r(1 - \pi^{\ast}_{\mathrm{pre}})\big) - \frac{1}{2} \sigma^{2} (\pi^{\ast}_{\mathrm{pre}})^{2} + \lambda \log(1 - \pi^{\ast}_{\mathrm{pre}}).
    \end{equation*}
\end{proposition}

\begin{proof}
    Plugging
    \begin{equation*}
        \partial_{t} V(t, x) = \lambda e^{-\lambda t} \log(x) + \dot{G}(t), \;
        \partial_{x} V(t, x) = e^{-\lambda t} (1/x), \;
        \partial_{xx} V(t, x) = -e^{-\lambda t} (1/x^{2})
    \end{equation*}
    into HJB Equation~\eqref{EQUATION:HJB_PDE}, we obtain
    \begin{equation*}
        \begin{split}
            \lambda e^{-\lambda t} \log(x) &+ \dot{G}(t) - \big(\mu \pi^{\ast}_{\mathrm{pre}} + r(1 - \pi^{\ast}_{\mathrm{pre}})\big) e^{-\lambda t} \\
            &+ \frac{1}{2} \sigma^{2} (\pi^{\ast}_{\mathrm{pre}})^{2} e^{-\lambda t} - \lambda e^{-\lambda t} \log(1 - \pi^{\ast}_{\mathrm{pre}}) - \lambda e^{-\lambda t} \log(x) = 0
        \end{split}
    \end{equation*}
    or, equivalently,
    \begin{equation}
        \label{EQUATION:VALUE_FUNCTION_G_ODE}
        \dot{G}(t) = C e^{-\lambda t} \text{ for } t \geq 0
    \end{equation}
    with
    \begin{equation}
        \label{EQUATION:VALUE_FUNCTION_CONSTANT_C_ORIGINAL}
        C = \big(\mu \pi^{\ast}_{\mathrm{pre}} + r(1 - \pi^{\ast}_{\mathrm{pre}})\big) - \frac{1}{2} \sigma^{2} (\pi^{\ast}_{\mathrm{pre}})^{2} + \lambda \log(1 - \pi^{\ast}_{\mathrm{pre}}).
    \end{equation}

    Integrating Equation~\eqref{EQUATION:VALUE_FUNCTION_G_ODE}, we obtain $G(t) = -\frac{C}{\lambda} e^{-\lambda t} + c$ for some constant $c$. On the strength of Equation~\eqref{EQUATION:HJB_TC}, we have the terminal condition
    \begin{equation*}
        -\frac{C}{\lambda} e^{-\lambda T} + c = - r\Big(T - \frac{1 - e^{-\lambda T}}{\lambda}\Big),
    \end{equation*}
    whence $c = \frac{C}{\lambda} e^{-\lambda T} - r\Big(T - \frac{1 - e^{-\lambda T}}{\lambda}\Big)$ and
    \begin{equation*}
        G(t) = \frac{C}{\lambda} \big(e^{-\lambda T} - e^{-\lambda t}\big) - r\Big(T - \frac{1 - e^{-\lambda T}}{\lambda}\Big)
    \end{equation*}
    for $t \geq 0$ and $x > 0$.
\end{proof}

Interestingly, the following observations can be made.
For $\mu - r < \sigma^{2}$, using L'H\^{o}pital's rule, $V(t, x; \lambda)$ can be shown to converge to the classical Merton value function, i.e.,
\begin{equation*}
    \lim_{\lambda \searrow 0} V(t, x; \lambda) = -\log(x) - \Big(r + \frac{(\mu - r)^{2}}{\sigma^{2}}\Big) (T - t).
\end{equation*}
In contrast, for $\mu - r >\sigma^{2}$,
\begin{equation*}
    \lim_{\lambda \searrow 0} V(t, x; \lambda) = -\log(x) - \Big(\mu - \frac{1}{2}\sigma^{2}\Big) (T - t).
\end{equation*}
since $\lim_{\lambda \searrow 0} \pi^{\ast}_{\mathrm{pre}}(\lambda) = 1$ and $\lim_{\lambda \searrow 0} C^{\ast}(\lambda) = \mu - \frac{1}{2} \sigma^{2}$.

Observing
\begin{equation*}
    V(t, x), \partial_{t} V(t, x), \partial_{xx} V(t, x) \in C^{0}([0, T] \times (0, \infty), \mathbb{R})
\end{equation*}
with
\begin{equation*}
    \partial_{xx} V(t, x) = e^{-\lambda t} x^{-2} > 0 \text{ and } \pi^{\ast}_{t} \in \mathcal{A}_{tx},
\end{equation*}
the conditions of Verification Theorem~\cite[Theorem~3]{AlPe2021} are satisfied, furnishing the following optimality result.

\begin{proposition}
    \label{PROPOSITION:FINAL}

    Let the wealth process be as in Equation~\eqref{EQUATION:ACTUAL_WEALTH}.
    \begin{itemize}
        \item The Merton ratio $\pi^{\ast}_{t} = \pi^{\ast}_{\mathrm{pre}} \mathds{1}_{\{t < \tau\}}$
        with
        $\pi^{\ast}_{\mathrm{pre}}$ given in Equation~\eqref{EQUATION:MERTON_RATIO} minimizes the negative expected utility $J(\pi)$ in Equation~\eqref{EQUATION:OBJECTIVE_FUNCTIONAL_CONTROLLED_DIFFUSION} over all admissible controls $\pi \in \mathcal{A}_{tx}$ and satisfies
        \begin{equation*}
            J(\pi^{\ast}) = \inf_{\pi \in \mathcal{A}_{tx}} J(\pi) = V(t, x).
        \end{equation*}

        \item The optimal wealth process, $W^{\ast}_{t}$, is given by Equations~\eqref{EQUATION:SDE_WEALTH}--\eqref{EQUATION:ACTUAL_WEALTH} with $\pi^{\ast}_{t}$ as above.
    \end{itemize}
\end{proposition}

\section{Example: Bombardier Stock Data Analysis}
\label{SECTION:EXAMPLE}

To illustrate our proposed methodology, we consider the B-rated stock of Bombardier Inc. (\texttt{BBD-B.TO}), a Canadian aerospace and transportation company, quoted by Toronto Stock Exchange. The annualized rating default probability is estimated as $\hat{\lambda} = 2.99\%$~\cite[Table 23]{SPG2022}. For simplicity, we use $\hat{\lambda} = 3\%$ in our calculations. Downloading a three years' worth of daily closing prices from~\formatdate{29}{1}{2021} to~\formatdate{29}{1}{2024} from Yahoo!~Finance~\cite{YF2024},
we compute the log-differenced returns
\begin{equation*}
    X_i = \log(S_{t_{i+1}}) - \log(S_{t_{i}}) \text{ for } i = 1, 2, \dots, n
\end{equation*}
with $n = 751$ furnishing the the annualized drift and volatility estimates (cf.~\cite{AnPo2023, PoAn2022})
\begin{equation}
    \label{EQUATION:MLE_ESTIMATES}
    \hat{\mu} = \frac{n_{\mathrm{trade}}}{n} \sum_{i=1}^{n} X_{i}, \quad
    \hat{\sigma} = \sqrt{n_{\mathrm{trade}}} \Big(\frac{1}{n - 1} \sum_{i = 1} X_{i}^{2} - \frac{n}{n - 1} \Big(\frac{1}{n} \sum_{i = 1}^{n} X_{i}\Big)^{2}\Big)^{1/2}
\end{equation}
where $n_{\mathrm{trade}} = 252$ stands for the number of trading days. See columns 1 and 2 in Table~\ref{TABLE:EXAMPLE}.

Starting with the logarithmic utility, the classical Merton's ratio $\hat{\pi}_{\mathrm{classic}}^{\ast}$ recommends to invest $101.12\%$ of capital in stock while only $74.32\%$ of capital is allocated in stock pre-default under our new ratio $\hat{\pi}_{\mathrm{pre}}^{\ast}$. In addition to offering a much less conservative investment strategy, classical Merton ratio is not even admissible (cf.~Equation~\eqref{EQUATION:ADMISSIBLE_CONTROLS_X_PROCESS}) in the presence of a sporadic bankruptcy as it renders the expected terminal utility equal negative infinity. Therefore, the difference between our framework and the classical framework of Merton is not merely quantitative, but qualitative.

\begin{table}[h!]
    \centering

    \caption{\label{TABLE:EXAMPLE}
    Estimated model parameters $\hat{\mu}$, $\hat{\sigma}$, $\hat{\lambda}$ (annualized)
    and resulting optimal allocations $\hat{\pi}_{\mathrm{classic}}^{\ast}$ (both pre- and post-default) according to Merton and $\hat{\pi}_{\mathrm{pre}}^{\ast}$ and $\hat{\pi}_{\mathrm{post}}^{\ast}$ derived in this paper.}

    \begin{tabular}{ccc|c|cc}
        \hline
        $\hat{\mu}$ & $\hat{\sigma}$ & $\hat{\lambda}$ & $\hat{\pi}_{\mathrm{classic}}^{\ast}$ & $\hat{\pi}_{\mathrm{pre}}^{\ast}$ & $\hat{\pi}_{\mathrm{post}}^{\ast}$ \\
        \hline \hline
        $40.27\%$ &  $59.05\%$ & $2.4\%$ & $101.12\%$ & $74.32\%$ & $0\%$ \\
        \hline
    \end{tabular}
\end{table}

\begin{figure}[h]
    \centering

    \begin{subfigure}[b]{0.31\textwidth}
        \includegraphics[width=\textwidth]{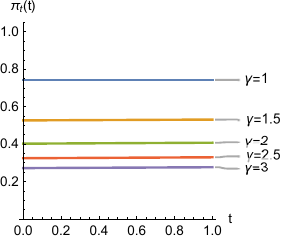}
        \caption{\label{FIGURE:MERTON_RATIO_NON_MYOPIC_T_EQUAL_1}
        Maturity $T = 1$ year.}
    \end{subfigure}
    ~
    \begin{subfigure}[b]{0.31\textwidth}
        \includegraphics[width=\textwidth]{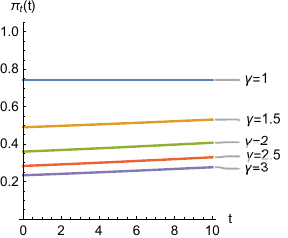}
        \caption{\label{FIGURE:MERTON_RATIO_NON_MYOPIC_T_EQUAL_10}
        Maturity $T = 10$ years.}
    \end{subfigure}
    ~
    \begin{subfigure}[b]{0.31\textwidth}
        \includegraphics[width=\textwidth]{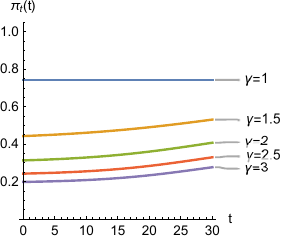}
        \caption{\label{FIGURE:MERTON_RATIO_NON_MYOPIC_T_EQUAL_30}
        Maturity $T = 30$ years.}
    \end{subfigure}

    \caption{
    \label{FIGURE:MERTON_RATIO_NON_MYOPIC}
    Myopic pre-default Merton ratio (for logarithmic utility) and non-myopic pre-default Merton ratios for various exponents $\gamma$ for $T = 1$ year (left), $T = 10$ years (center) and $T = 30$ years (right).}
\end{figure}

We now turn to isoelastic utility with $\gamma \neq 1$ and analyze the optimal allocation for the same stock scenario. Figure~\ref{FIGURE:MERTON_RATIO_NON_MYOPIC} plots the optimal pre-default stock weight $\hat{\pi}_{\mathrm{pre}, t}^{\ast}$ vs.~time $t$ for the time horizon $T = 10$ years (left panel) and $T = 30$ years (right panel) for five different $\gamma$ values ($\gamma = 1$ corresponding to the logarithmic utility). As can be seen, for $\gamma\neq 1$, the weights are non-myopic, which is more strongly pronounced for larger $\gamma$'s and $T$'s.

As the left panel of Figure~\ref{FIGURE:MERTON_RATIO_NON_MYOPIC} suggests, the optimal pre-default Merton ratio is nearly linear for shorter maturities $T$. In particular, for $T = 1$, the weight is almost constant. The general linearization equation for $\pi_{t}$ around $t = T$ can be computed as
\begin{equation*}
    \pi_{t} \approx \pi_{T} - (T - t) \frac{\mathrm{d}}{\mathrm{d}t} \pi_{t}\big|_{t = T}
    \equiv \pi_{T} - (T - t) \frac{\gamma \sigma^2
    (1 -\pi_T) \big(\alpha - \pi_T\big)
    \big(\tfrac{1}{2} \pi_T^2 - \beta \pi_T + \eta\big)} {\pi_T  - \beta}
\end{equation*}
using ODE~\eqref{EQUATION:ODE_FOR_PI_T}. The approximation is quite accurate for a wide choice of practically relevant parameter values, especially for small $T$ and $\gamma$ close to 1. Linearized equations (around $T$) for our example are provided in Table~\ref{TABLE:EXAMPLE_LINEARIZATION}.

\begin{table}[h!]
    \centering

    \caption{\label{TABLE:EXAMPLE_LINEARIZATION}
    Linearized time-dependent stock optimal stock weights for small maturities $T$.}

    \footnotesize
    \begin{tabular}{cccccc}
        \hline
        $\gamma$ & $1.0$ & $1.5$ & $2.0$ & $2.5$ & $3.0$ \\
        \hline \hline
        $\hat{\pi}^{\ast}_{\mathrm{pre}, t}$ & $0.74333$ & $0.53132 {-} 0.00449 (T {-} t)$ &
        $0.40766 {-} 0.00511 (T {-} t)$ & $0.32974 {-} 0.00489 (T {-} t)$ & $0.27657 {-} 0.00451 (T {-} t)$ \\
        \hline
    \end{tabular}
\end{table}

\section{Summary and Conclusions}

In this work, we presented an explicit solution to the optimal wealth allocation problem between a stock and a cash bond assuming the stock bankruptcy is an exogenous event with a constant intensity $\lambda$. To the best of our knowledge, neither this type of problems, nor the solution we developed for the $\log$-utility have previously appeared in the literature. The stochastic integral derivation for this type of processes appears to be new as well. Our solution has an interesting feature that one should never allocates all money in a stock that can go bankrupt to hedge against an entire loss of one's money. This seems to correlate with the investor behavior in real market. The framework we developed is not limited to bankruptcy events only. In fact, stochastic control problems with regime or security changes can be solve using a similar method. These examples include convertible bonds (i.e., bonds with a conversion feature), bonds with embedded options assuming the probability of embedded options are governed by an exogenous probability, mortgage bonds, etc. Our framework can also be applied in the context of switching between economic regimes caused by exogenous events in the market. We intend to investigate these and other scenarios in subsequent work.



\section*{Data Availability Statement}
The Bombardier Inc.~stock price dataset analyzed in Section~\ref{SECTION:EXAMPLE} is available at~\url{https://github.com/mpokojovy/MertonAbsorbing/blob/main/BOMB.csv}.

\section*{Supplemental Online Material}
The following supplementary materials are available online at: \\
\url{https://github.com/mpokojovy/MertonAbsorbing}:
\begin{enumerate}
    \item Python code to reproduce Figure~\ref{FIGURE:MERTON_RATIO}.

    \item \texttt{R} code implementing the results of Section~\ref{SECTION:EXAMPLE}.
\end{enumerate}

\appendix

\section{Auxiliary Results}
\label{APPENDIX:AUXILIARY_RESULTS}

\subsection{Results for Section~\ref{SECTION:HEURISTIC_DERIVATION}}
\label{APPENDIX:AUXILIARY_RESULTS_GENERAL_GAMMA}

\begin{lemma}
    \label{LEMMA:CONCAVITY_APPENDIX}
    The function $h(t, \pi)$ defined in Equation~\eqref{EQUATION:HJB_POWER_UTILITY_H_FUNCTION_REDUCED} is concave in $\pi$ for
    \begin{itemize}
        \item $\pi \in (-\infty, \alpha) \cup (\beta, 1)$ and $t \in [0, T]$ if $\alpha < 1$,

        \item $\pi \in (-\infty, 1)$ and $t \in [0, T]$ if $\alpha \geq 1$.
    \end{itemize}
\end{lemma}
\begin{proof}
    Differentiation yields
    \begin{equation*}
        \partial_{\pi}^{2} h(t, \pi) =
        -(\gamma+1) \lambda (1 - \pi)^{-\gamma-1} \frac{\pi - \beta}{\pi - \alpha}.
    \end{equation*}
    A simple calculation suggests
    $\beta - \alpha = \frac{1 - \alpha}{\gamma + 1} := \delta$ or $\beta = 1 - \gamma \delta$.
    Thus, if $\alpha = 1$, we have $\delta = 0$ and $\beta = 1$, whence $\partial_{\pi}^{2} h(t, \pi) \leq 0$.
    If $\alpha > 1$, then $\delta < 0$, and so $1 < \beta < \alpha$ so that $\partial_{\pi}^{2} h(t, \pi)$ for all $\pi < 1$.
    If $\alpha < 1$, then $\delta > 0$, and so $\alpha < \beta < 1$.
    Thus, the function $h$ is concave in $\pi$ for $\pi < \min\{\alpha, \beta\}$ or $\pi > \max\{\alpha, \beta\}$ and $t \in [0, T]$. Taking into account the relation between $\alpha$ and $\beta$, the claim follows.
\end{proof}

\begin{proposition}
    \label{PROPOSITION:ODE_SOLUTION_PLUS_MAXIMIZER}

    The ODE in Equation~\eqref{EQUATION:ODE_FOR_PI_T} subject to terminal condition in Equation~\eqref{piTEq} possesses a unique solution $\pi_{t}$ for $t \in [0, T]$ such that $\pi_{t} = \max_{\pi < 1} h(t, \pi)$. Moreover, $\pi_{t} \leq \alpha$ holds for all $t \in [0, T]$. (Recall that $\alpha$ defined in Equation~\eqref{EQUATION:CONSTANTS_ALPHA_BETA_ETA} is the classical Merton ratio.)
\end{proposition}

\begin{proof}
    Consider the function $\phi$ defined in Equation~\eqref{piTEq}.
    On the interval $(-\infty, 1)$, the function $\phi$ is monotonically increasing and concave, and has a vertical asymptote at $\pi = 1$. Thus, there exists a unique root, denoted as $\pi^{\ast}$, of $\phi$ on this interval, since $\phi(\pi) \to -\infty$ as $\pi \to -\infty$, and $\phi(\pi) \to \infty$ as $\pi \to 1$. Moreover, $\pi^{\ast} < \alpha$ since $\phi(\pi) < 0$ for $\pi \in (\alpha, 1)$. The terminal condition $\phi(\pi_{T}) = 0$ necessitates that $\pi_{T} = \pi^{\ast}$.

    The usual barrier method implies then that $\pi_{t} \leq \alpha$ holds true on some maximum existence interval since $\kappa$ is locally Lipschitz. As long as $\pi_{T}$ is ``squeezed'' between two zeros of the right-hand side $\kappa$ of ODE~\eqref{EQUATION:ODE_FOR_PI_T}, the solution must exist for all times and stay between those two zeros. Should $\pi_{T}$ lie below the smallest zero of $\kappa$, calculating $\limsup_{\pi \to -\infty} \kappa(\pi) \leq 0$, we conclude that $\kappa(\pi_{T}) \leq 0$ and, thus, $\pi_{t}$ is monotonically decreasing. Therefore, the solution exists for all $t \leq T$ and converges to the smallest zero of $\kappa$ as $t \to -\infty$.

    Thus, we have shown the solution $\pi_{t}$ exists for all $t \in [0, T]$, is bounded by $\alpha$ and satisfies $\pi_{t} = \max_{\pi < 1} h(t, \pi)$ on the strength of Lemma~\ref{LEMMA:CONCAVITY_APPENDIX}.
\end{proof}

\subsection{Results for Section~\ref{SECTION:RIGOROUS_DERIVATION}}

\begin{lemma}
    \label{LEMMA:HJB_OPERATOR_MINIMIZATION}

    Let $g(t, x, \pi, p, A) = f(t, x, \pi) \cdot p + \frac{1}{2} \sigma^{2}(t, x, \pi) \cdot A - L(t, x, \pi)$ with
    \begin{align*}
        f(t, x, \pi) = \big(\mu \pi + r (1 - \pi)\big) x, \quad
        \sigma(t, x, \pi) = \sigma \pi x, \quad
        L(t, x, \pi) = \lambda e^{-\lambda t} \log\big((1 - \pi) x\big).
    \end{align*}
    Assuming $x > 0$ and $A > 0$,
    \begin{align*}
        \mathop{\operatorname{arg\,min}}_{\pi < 1} g(t, x, \pi, p, A) &= \pi^{\ast}(t, x, p, A), \\
        \inf_{\pi < 1} g(t, x, \pi, p, A) &= g\big(t, x, \pi^{\ast}(t, x, p, A), p, A\big)
    \end{align*}
    with {\small
    \begin{equation}
        \label{EQUATION:OPTIMAL_PROPORTION}
        \pi^{\ast}(t, x, p, A) =
        \frac{-(\mu - r) x p + \sigma^{2} x^{2} A {-} \sqrt{((\mu - r) x p + \sigma^{2} x^{2} A)^{2} + 4 \lambda e^{-\lambda t} (\sigma^{2} x^{2} A)}}{2 (\sigma^{2} x^{2} A)}.
    \end{equation}}
\end{lemma}

\begin{proof}
    The function
    \begin{align*}
        g(t, x, \pi, p, A) = \big(\mu \pi + r (1 - \pi)\big) x \cdot p + \frac{1}{2} \sigma^{2} \pi^{2} x^{2} \cdot A - \lambda e^{-\lambda t} \log\big((1 - \pi) x\big)
    \end{align*}
    is smooth with respect to all of its arguments $t \geq 0$, $x > 0$, $\pi < 1$, $p \in \mathbb{R}$ and $A > 0$. Since $\lim_{\pi \nearrow 0} g(t, x, \pi, p, A) = \infty$, no minimum can be attained in a vicinity of $\pi = 1$.

    Computing the derivative
    \begin{align*}
        \partial_{\pi} g(t, x, \pi, p, A) =
        (\mu - r) x p + \sigma^{2} \pi x^{2} A + \frac{\lambda e^{-\lambda t}}{1 - \pi},
    \end{align*}
    the Fermat first-order optimality equation $\partial_{\pi} g(t, x, \pi, p, A) = 0$ reads as
    \begin{equation*}
        \sigma^{2} \pi x^{2} A + (\mu - r) x p + \frac{\lambda e^{-\lambda t}}{(1 - \pi)} = 0
    \end{equation*}
    or, multiplying with $(\pi - 1)$,
    \begin{equation}
        \label{EQUATION:QUADRATIC_POLYNOMIAL_APPENDIX}
        P(\pi; \lambda t) \equiv \big(\sigma^{2} x^{2} A\big) \pi^{2} +
        \big((\mu - r) x p - \sigma^{2} x^{2} A\big) \pi -
        \big((\mu - r) x p + \lambda e^{-\lambda t}\big) = 0.
    \end{equation}

    Solving the former quadratic equation, we obtain two roots given by {\small
    \begin{align*}
        \pi_{1, 2} &= \frac{-\big((\mu {-} r) x p {-} \sigma^{2} x^{2} A\big) {\pm} \sqrt{\big((\mu - r) x p {-} \sigma^{2} x^{2} A\big)^{2} + 4 \big(\sigma^{2} x^{2} A\big) \big((\mu - r) x p + \lambda e^{-\lambda t}\big)}}{2 \big(\sigma^{2} x^{2} A\big)} \\
        &= \frac{-\big((\mu - r) x p - \sigma^{2} x^{2} A\big) \pm \sqrt{\big((\mu - r) x p + \sigma^{2} x^{2} A\big)^{2} + \lambda e^{-\lambda t} \big(\sigma^{2} x^{2} A\big)}}{2 \big(\sigma^{2} x^{2} A\big)}.
    \end{align*}}

    In view of $A > 0$, both roots are real and satisfy $\pi_{1} \leq \pi_{2}$. Moreover, since $P(\pi; \lambda, t) > P(\pi; 0, t)$ for $\lambda > 0$ and $t \geq 0$, we observe
    \begin{equation*}
        \pi_{1}(\lambda t) < \pi_{1}(0) \quad \text{ and } \quad \pi_{2}(\lambda t) > \pi_{2}(0) \quad \text{ for } \lambda > 0 \text{ and } t \geq 0.
    \end{equation*}
    Consider the following two cases.
    \begin{itemize}
        \item Assuming $(\mu - r) xp \leq \sigma^{2} Ax^{2}$, we obtain
        \begin{equation*}
            \pi_{1}(\lambda t) < \pi_{1}(0) \equiv -\frac{\mu - r}{\sigma^{2}} \Big(\frac{p}{Ax}\Big) \leq 1 \quad \text{ and } \quad \pi_{2}(\lambda t) > \pi_{2}(0) \equiv 1
        \end{equation*}
        for $\lambda > 0$ and $t \geq 0$
        and, thus, $\partial_{\pi} g(t, x, \pi, p, A) \leq 0$ for $\pi < \pi_{1}(\lambda t)$ and $\partial_{\pi} g(t, x, \pi, p, A) \geq 0$ for $\pi_{1}(\lambda t) < \pi < 1$.

        \item Otherwise, $(\mu - r) xp > \sigma^{2} Ax^{2}$, which implies
        \begin{equation*}
            \pi_{1}(\lambda t) < \pi_{1}(0) \equiv 1 \quad \text{ and } \quad \pi_{2}(\lambda t) \geq \pi_{2}(0) \equiv -\frac{\mu - r}{\sigma^{2}} \Big(\frac{p}{Ax}\Big) > 1
        \end{equation*}
        for $\lambda > 0$ and $t \geq 0$
        and, thus, again $\partial_{\pi} g(t, x, \pi, p, A) \leq 0$ for $\pi < \pi_{1}(\lambda t)$ and $\partial_{\pi} g(t, x, \pi, p, A) \geq 0$ for $\pi_{1}(\lambda t) < \pi < 1$.
    \end{itemize}
    Therefore, in either scenario, the global minimum of $g = g(\pi)$ is attained at $\pi = \pi_{1}(\lambda t)$ for $\lambda > 0$, $t \geq 0$ and $x > 0$ assuming $A > 0$.
\end{proof}


\begin{thebibliography}{}

\bibitem[A\"{i}t-Sahalia and Hurd, 2016]{SH2015}
A\"{i}t-Sahalia, Y. and Hurd, T.~R. (2016).
\newblock Portfolio choice in markets with contagion.
\newblock {\em Journal of Financial Econometrics}, 14:1--28.

\bibitem[Albosaily and Pergamenchtchikov, 2021]{AlPe2021}
Albosaily, S. and Pergamenchtchikov, S. (2021).
\newblock Optimal investment and consumption for multidimensional spread
  financial markets with logarithmic utility.
\newblock {\em Stats}, 4:1012--1026.

\bibitem[Anum and Pokojovy, 2023]{AnPo2023}
Anum, A.~T. and Pokojovy, M. (2023).
\newblock A hybrid method for density power divergence minimization with
  application to robust univariate location and scale estimation.
\newblock {\em Communications in Statistics -- Theory and Methods},
  53(14):5186–5209.

\bibitem[Azevedo et~al., 2014]{AzPiWe2014}
Azevedo, N., Pinheiro, D., and Weber, G.-W. (2014).
\newblock Dynamic programming for a {M}arkov-switching jump–diffusion.
\newblock {\em Journal of Computational and Applied Mathematics}, 267:1--19.

\bibitem[Bielecki and Jang, 2006]{BJ}
Bielecki, T.~R. and Jang, I. (2006).
\newblock Portfolio optimization with a defaultable security.
\newblock {\em Asia Pacific Financial Markets}, 13:113–127.

\bibitem[Brennan et~al., 1997]{BSL1997}
Brennan, M., Schwartz, E., and Lagnado, R. (1997).
\newblock Strategic asset allocation.
\newblock {\em Journal of Economic Dynamics and Control}, 1:1377--1403.

\bibitem[Capponi and Figueroa~Lopez, 2011]{CFL}
Capponi, A. and Figueroa~Lopez, J. (2011).
\newblock Dynamic portfolio optimization with a defaultable security and regime
  switching.
\newblock {\em Mathematical Finance}, 24(2):207--249.

\bibitem[Capponi and Larsson, 2015]{CL}
Capponi, A. and Larsson, M. (2015).
\newblock Default and systemic risk in equilibrium.
\newblock {\em Mathematical Finance}, 25(1):51--76.

\bibitem[Emmer and Kl\"{u}ppelberg, 2004]{EK}
Emmer, S. and Kl\"{u}ppelberg, C. (2004).
\newblock Optimal portfolios when stock prices follow an exponential {L}ev\'{y}
  process.
\newblock {\em Finance and Stochastics}, 8:17--44.

\bibitem[Fei, 2014]{F2014}
Fei, W. (2014).
\newblock Optimal control of uncertain stochastic systems with {M}arkovian
  switching and its application to portfolio decisions.
\newblock {\em Cybernetics and Systems}, 45(1):69--88.

\bibitem[Fleming and Rishel, 1975]{FleRi1975}
Fleming, W.~H. and Rishel, R.~W. (1975).
\newblock {\em Deterministic and Stochastic Optimal Control}, volume~1 of {\em
  Applications of Mathematics}.
\newblock Springer, New York, NY.

\bibitem[Fleming and Soner, 2006]{FleSo2006}
Fleming, W.~H. and Soner, H.~M. (2006).
\newblock {\em Controlled Markov Processes and Viscosity Solutions}, volume~25
  of {\em Stochastic Modelling and Applied Probability}.
\newblock Springer Science \& Business Media, New York, NY.

\bibitem[Fouque et~al., 2017]{FoSiZa2017}
Fouque, J.-P., Sircar, R., and Zariphopoulou, T. (2017).
\newblock Portfolio optimization and stochastic volatility asymptotics.
\newblock {\em Mathematical Finance}, 27(3):704--745.

\bibitem[Goll and Kallsen, 2000]{GoKa2000}
Goll, T. and Kallsen, J. (2000).
\newblock Optimal portfolios for logarithmic utility.
\newblock {\em Stochastic Processes and their Applications}, 89:31--48.

\bibitem[Karatzas et~al., 1986]{KLSS1985}
Karatzas, I., Lehoczky, J., Sethi, S., and Shreve, S. (1986).
\newblock Explicit solution of a general consumption/investment problem.
\newblock {\em Mathematics of Operations Research}, 11(2):261--294.

\bibitem[Klinger, 1970]{Kli1970}
Klinger, A. (1970).
\newblock Control with stochastic stopping time.
\newblock {\em International Journal of Control}, 11:541--549.

\bibitem[Kloeden and Platen, 1992]{KlPl1992}
Kloeden, P.~E. and Platen, E. (1992).
\newblock {\em Numerical Solution of Stochastic Differential Equations},
  volume~23 of {\em Stochastic Modelling and Applied Probability}.
\newblock Springer, Berlin, Heidelberg.

\bibitem[Kopeliovich, 2023]{KOP2023}
Kopeliovich, Y. (2023).
\newblock Optimal control problems for stochastic processes with absorbing
  regime.
\newblock {\em Journal of Stochastic Analysis}, 4(4):Article 6.

\bibitem[Kraemer et~al., 2022]{SPG2022}
Kraemer, N.~W., Palmer, J., Richhariya, N.~M., Iyer, S., Fernandes, L., Meher,
  A., and Pranshu, S. (2022).
\newblock {Default, Transition, and Recovery: 2021 Annual Global Corporate
  Default and Rating Transition Study}.
\newblock
  \url{https://www.spglobal.com/ratings/en/research/articles/220413-default-transition-and-recovery-2021-annual-global-corporate-default-and-rating-transition-study-12336975}.

\bibitem[Leland, 1994]{Le1994}
Leland, H. (1994).
\newblock Corporate debt value, bond covenants, and optimal capital structure.
\newblock {\em Journal of Finance}, 49:1213--1252.

\bibitem[Leland et~al., 2001]{LeGo2001}
Leland, H., Goldstein, R., and Ju, N. (2001).
\newblock An {EBIT}-based model of optimal capital structure.
\newblock {\em Journal of Business}, 74(4):483--512.

\bibitem[Lukashiv et~al., 2023]{LuLiMaGoNa2023}
Lukashiv, T., Litvinchuk, Y., Malyk, I., Golebiewska, A., and Nazarov, P.
  (2023).
\newblock Stabilization of stochastic dynamical systems of a random structure
  with {M}arkov switches and {P}oisson perturbations.
\newblock {\em Mathematics}, 11(3):582.

\bibitem[Medhi, 1984]{Me1984}
Medhi, J. (1984).
\newblock {\em Stochastic Processes}.
\newblock New Age International Publishers, New Delhi, 2nd edition.

\bibitem[Merton, 1969]{Mer1969}
Merton, R. (1969).
\newblock Lifetime portfolio selection under uncertainty: the continuous-time
  case.
\newblock {\em The Review of Economics and Statistics}, 51:247--257.

\bibitem[Merton, 1973]{Mer1973}
Merton, R. (1973).
\newblock An intertemporal capital asset pricing model.
\newblock {\em Econometrica}, 41:867--877.

\bibitem[Platen and Rendek, 2009]{PlaRe2009}
Platen, E. and Rendek, R. (2009).
\newblock Exact scenario simulation for selected multi-dimensional stochastic
  processes.
\newblock {\em Communications on Stochastic Analysis}, 3(3):443--465.

\bibitem[Pokojovy and Anum, 2022]{PoAn2022}
Pokojovy, M. and Anum, A.~T. (2022).
\newblock A fast initial response approach to sequential financial
  surveillance.
\newblock In Han, H. and Baker, E., editors, {\em The Recent Advances in
  Transdisciplinary Data Science}, pages 19--33, Cham. Springer Nature
  Switzerland.

\bibitem[Pun and Wong, 2015]{PuWo2015}
Pun, C.~S. and Wong, H.~Y. (2015).
\newblock Robust investment‑reinsurance optimization with multiscale
  stochastic volatility.
\newblock {\em Insurance: Mathematics and Economics}, 62:245‑256.

\bibitem[{The Economist. Buttonwood (Online Column)}, 2023]{Bu}
{The Economist. Buttonwood (Online Column)} (2023).
\newblock How to avoid a common investment mistake.
\newblock
  \url{https://www.economist.com/finance-and-economics/2023/09/21/how-to-avoid-a-common-investment-mistake}.

\bibitem[Wachter, 2002]{Wa2002}
Wachter, J. (2002).
\newblock Portfolio and consumption decisions under mean-reverting returns: an
  exact solution for complete markets.
\newblock {\em Journal of Financial And Quantitative Analysis}, 37:63--91.

\bibitem[Wei et~al., 2020]{WeShZh2020}
Wei, J., Shen, Y., and Zhao, Q. (2020).
\newblock Portfolio selection with regime-switching and state-dependent
  preferences.
\newblock {\em Journal of Computational and Applied Mathematics}, 365:112361.

\bibitem[{Yahoo!~Finance}, 2024]{YF2024}
{Yahoo!~Finance} (2024).
\newblock {Bombardier Inc.~(\texttt{BBD-B.TO}) Historical Prices}.
\newblock \url{https://finance.yahoo.com/quote/BBD-B.TO/history}.

\end{thebibliography}

\end{document}